\documentclass[12pt,reqno]{amsart}

\usepackage[T1]{fontenc}
\usepackage[english]{babel}
\usepackage[all]{xy}
\usepackage{graphicx}
\usepackage{verbatim,amssymb}

\newtheorem{thm}{Theorem}[section]
\newtheorem{cor}[thm]{Corollary}
\newtheorem{lem}[thm]{Lemma}

\theoremstyle{definition}
\newtheorem{defn}[thm]{Definition}
\theoremstyle{remark}
\newtheorem{rem}[thm]{Remark}
\newtheorem{eg}[thm]{Example}

\DeclareMathOperator*{\PLUS}{\oplus}
\allowdisplaybreaks
\begin{document}
\allowdisplaybreaks
\title[Trace formulae for graph Laplacians]{Trace formulae for graph Laplacians with applications to recovering matching conditions}
\author{Yulia Ershova}
\address{Institute of Mathematics, National Academy of Sciences of Ukraine. 01601 Ukraine, Kiev-4,
3, Tereschenkivska st.}
\email{julija.ershova@gmail.com}
\thanks{The first authors' work was partially supported by the grant DFFD F40.1/008.}
\author{Alexander V. Kiselev}
\address{Department of Higher Mathematics and Mathematical Physics,
St. Petersburg State University, 1 Ulianovskaya Street,
St. Petersburg, St. Peterhoff 198504 Russia}
\email{alexander.v.kiselev@gmail.com}
\thanks{The second authors' work was partially supported by the grants RFBR 11-01-90402-Ukr\_f\_a and RFBR 12-01-00215-a.}

\begin{abstract}
Graph Laplacians on finite compact metric graphs are considered under the assumption that the matching conditions at the graph vertices are of either $\delta$ or $\delta'$ type.
In either case, an infinite series of trace formulae
which link together two different graph Laplacians provided that their spectra coincide is derived.
Applications are given to the problem of reconstructing matching conditions for a graph Laplacian based on its spectrum.
\end{abstract}

\maketitle
\specialsection{Introduction}
A graph Laplacian is a particular case of a quantum graph, i.e., a metric graph $\Gamma$ and an associated second-order differential operator
acting on the Hilbert space $L^2(\Gamma)$ of square summable functions on the graph with an additional assumption that the functions belonging to
the domain of the operator are coupled by certain matching conditions at the graph vertices. These matching conditions reflect the graph
connectivity and usually are assumed to guarantee self-adjointness of the operator. Recently these operators have attracted a considerable interest
of both physicists and mathematicians due to a number of important physical applications, e.g., to the study of quantum wavequides. Extensive literature
on the subject is surveyed in, e.g., \cite{Kuchment}.

In the situation of a graph Laplacian the above-mentioned second-order differential operator is simply the operator of negative second derivative.

The present paper is devoted to the study of the inverse spectral problem for graph Laplacians on finite compact metric graphs.
One might classify the possible inverse problems for graph Laplacians in the following way.
\begin{itemize}
\item[(i)] Given spectral data and the matching conditions (usually one assumes standard matching conditions, see below), to reconstruct the metric graph;
\item[(ii)] Given the metric graph and spectral data, to reconstruct the matching conditions.
\end{itemize}
There exists an extensive literature devoted to the problem \emph{(i)}. To name just a few, we
would like to mention the pioneering works \cite{Roth,Smil1,Smil2} and later contributions
\cite{Kura1,Kura2,Kura3,Kostrykin}. These papers utilize an approach to the problem \emph{(i)} based on the so-called
trace formula which relates the spectrum of the quantum graph to the set of closed paths on the underlying metric graph.

On the other hand, the problem \emph{(ii)} has to the best of our knowledge surprisingly attracted much less interest. After being mentioned in \cite{Kura2}, it
was treated in \cite{Avdonin}, but only in the case of star graphs.

The present paper is devoted to the analysis of the same problem \emph{(ii)}. Unlike \cite{Avdonin}, we consider the case of a general connected compact finite metric graph
(in particular, this graph is allowed to possess cycles and loops), but only for two classes of matching conditions at the graph vertices, namely,
 in either the case of $\delta$ type
matching conditions at the vertices or the case of $\delta'$ type matching conditions (see Section 2 for definitions). The methods and mathematical
apparatus applied by us in both cases are identical, but the results prove to be somewhat different. The named two classes singled out by us
prove to be physically viable \cite{Exner1, Exner2}.

In contrast to \cite{Avdonin}, where the spectral data used in order to reconstruct the matching conditions it taken to be the Weyl-Titchmarsh M-function (or
Dirichlet-to-Neumann map) of the graph boundary, we use the spectrum of graph Laplacian (counting multiplicities) as the data known to us from the outset.

The approach suggested is based on the celebrated theory of boundary triples \cite{Gor}. Explicit construction a generalized Weyl-Titchmarsh M-function for a properly chosen
maximal (adjoint to a symmetric, which we refer to as \emph{minimal}) operator allows us to reduce the study of the spectrum of a graph Laplacian to the study of ``zeroes''
of the corresponding finite-dimensional analytic matrix function.
In order to achieve this goal, we surely have to construct an M-function for the whole graph rather than consider the Dirichlet-to-Neumann map pertaining
to the graph boundary.
On this path we are then able to derive an infinite series of trace formulae
which link together two different graph Laplacians provided that their spectra coincide. These trace formula surprisingly only involve the (diagonal) matrices
of coupling constants (i.e., constants appearing in matching conditions) and the diagonal matrix of the vertex valences of the graph $\Gamma$.

We would like to point out that the approach suggested seems to be applicable to the analysis of more general differential operators on a given metric graph, most
notably, of Schr\"odinger operators. We leave this question aside for the time being as we plan to make it a subject of a forthcoming publication.

The paper is organized as follows.

Section 2 introduces the notation and contains a brief summary of the material on the boundary triples used by us in the sequel. We continue by providing an explicit  ``natural'' form of
the Weyl-Titchmarsh M-function for the case of $\delta$ type matching conditions ($\delta'$ type, respectively), suitable for our goal. We further pay special attention to the
problem of simplicity of our minimal operator, which turns out to be equivalent to the question of whether the M-function together with the matrix of coupling constants
accounts for all of the spectrum of graph Laplacian or not.

Section 3 contains our main result, i.e., the trace formulae for graph Laplacians with $\delta$ type ($\delta'$ type, respectively) matching conditions. In this
Section we also draw certain corollaries from the trace formulae obtained pertaining to the inverse spectral problem for graph Laplacians in the setting \emph{(ii)}.

\specialsection{Boundary triples approach}

\subsubsection*{Definition of the Laplacian on a quantum graph}

In order to define the quantum Laplacian, i.e., the Laplace operator on a quantum graph, we begin with the following

\begin{defn}
We call $\Gamma=\Gamma(\mathbf{E_\Gamma},\sigma)$ a finite compact metric graph, if it is a collection of a finite non-empty set
$\mathbf{E_\Gamma}$ of finite closed intervals  $\Delta_j=[x_{2j-1},x_{2j}]$,
$j=1,2,\ldots, n$, called \emph{edges}, and of a partition
$\sigma$ of the set of endpoints $\{x_k\}_{k=1}^{2n}$ into $N$ classes, $\mathbf{V_\Gamma}=\bigcup^N_{m=1} V_m$. The equivalence classes
 $V_m$, $m=1,2,\ldots,N$ will be called \emph{vertices} and the number of elements belonging to the set $V_m$ will be called the \emph{valence} of the vertex
$V_m$.
\end{defn}

With a finite compact metric graph $\Gamma$ we associate the Hilbert space
$$L_2(\Gamma)=\PLUS_{j=1}^n L_2(\Delta_j).$$
This Hilbert space obviously doesn't feel the connectivity of the graph, being the same for each graph with the same number of edges of the same lengths.

In what follows, we single out two natural \cite{Exner1} classes of so-called \emph{matching conditions} which lead to a properly defined self-adjoint operator
on the graph $\Gamma$, namely, the matching conditions of $\delta$ and $\delta'$ types. In order to describe these, we will introduce the following notation.
For a smooth enough function $f\in L_2(\Gamma)$, we will use throughout the following definition of the normal derivative on a finite compact metric graph:
$$\partial_n f(x_j)=\left\{ \begin{array}{ll} f'(x_j),&\mbox{ if } x_j \mbox{ is the left endpoint of the edge},\\
-f'(x_j),&\mbox{ if } x_j \mbox{ is the right endpoint of the edge.}
\end{array}\right.$$

\begin{defn}\label{def_matching} If $f\in \PLUS_{j=1}^n W_2^2 (\Delta_j)$ and $\alpha_m$ is a complex number (referred to below as a coupling constant),
\begin{itemize}
\item[\textbf{($\delta$)}] the condition of continuity of the function $f$
through the vertex $V_m$ (i.e., $f(x_j)=f(x_k)$ if $x_j,x_k\in V_m$) together with the condition
$$
\sum_{x_j \in V_m} \partial _n f(x_j)=\alpha_m f(V_m)
$$
is called $\delta$-type matching at the vertex $V_m$;
\item[\textbf{($\delta'$)}] the condition of continuity of the normal derivative $\partial_n f$
through the vertex $V_m$ (i.e., $\partial_n f(x_j)=\partial_n f(x_k)$ if $x_j,x_k\in V_m$) together with the condition
$$
\sum_{x_j \in V_m} f(x_j)=\alpha_m \partial_n f(V_m)
$$
is called $\delta'$-type matching at the vertex $V_m$;
\end{itemize}
\end{defn}

\begin{rem}
Note that the $\delta$-type matching condition in a particular case when $\alpha_m=0$ reduces to the so-called standard, or Kirchhoff, matching condition
at the vertex $V_m$.
Note also that at the graph boundary (i.e., at the set of vertices of valence equal to 1) the $\delta$- and $\delta'$-type conditions reduce to the usual 3rd type
one, whereas the standard matching conditions lead to the Neumann condition at the graph boundary.
\end{rem}

We are all set now to define the graph Laplacian (i.e., the Laplace operator on a graph) on the graph $\Gamma$ with $\delta$- or $\delta'$-type matching conditions.
\begin{defn}\label{def_Laplace}
The graph Laplacian $A$ on a graph $\Gamma$ with $\delta$-type ($\delta'$-type, respectively) matching conditions is the operator of negative second derivative
in the Hilbert space $L_2(\Gamma)$ on the domain of functions belonging to the Sobolev space $\PLUS_{j=1}^n W_2^2(\Delta_j)$ and satisfying $\delta$-type
($\delta'$-type, respectively) matching conditions at every vertex $V_m$, $m=1,2,\dots,N.$
\end{defn}

\begin{rem}
Note that the matching conditions reflect the graph connectivity: if two graphs with the same edges have different topology, the resulting operators are different.
\end{rem}

Provided that all coupling constants $\alpha_m$, $m=1\dots N$, are real, it is easy to verify that the Laplacian $A$ is a self-adjoint operator in the
Hilbert space $L_2(\Gamma)$ \cite{Exner1,KostrykinS}. Throughout the present paper, we are going to consider this self-adjoint situation only, although it has
to be noted that the approach
developed can be used for the purpose of analysis of the general non-self-adjoint situation as well.

Clearly, the self-adjoint operator thus defined on a finite compact metric graph has purely discrete spectrum that might accumulate to $+\infty$ only. In order
to ascertain this, one only has to note that the operator considered is a finite-dimensional perturbation in the resolvent sense of the direct sum of Sturm-Liouville
operators on the individual edges.

\begin{rem}
Note that w.l.o.g. each edge $\Delta_j$ of the graph $\Gamma$ can be considered to be an interval $[0,l_j]$, where $l_j=x_{2j}-x_{2j-1}$, $j=1\dots n$ is
the length of the corresponding edge. Indeed, performing the corresponding linear change of variable one reduces the general situation to the one where
all the operator properties depend on the lengths of the edges rather than on the actual edge endpoints.
\end{rem}

We now pass over to the main subject of the present paper, i.e., to the derivation of an infinite series of trace formulae for the
graph Laplacian with $\delta$- or $\delta'$ matching conditions at the vertices. In order to do so, we will first need to establish an explicit
formula for the generalized Weyl-Titchmarsh M-function of the operator considered. The most elegant and straightforward way to do so is in our view
by utilizing the apparatus of boundary triples developed in \cite{Gor,Ko1,Koch,DM}. We briefly recall the results
essential for our work.

\subsubsection*{Boundary triplets and the Weyl-Titchmarsh matrix M-function}

Suppose that $A_{min}$ is a symmetric densely defined closed
linear operator acting in the Hilbert space $H$ ($D(A_{min})\equiv
D_{A_{min}}$ and $R(A_{min})\equiv R_{A_{min}}$ denoting its domain and
range respectively; $D(A_{max})\equiv D_{A_{max}}$,
$R(A_{max})\equiv R_{A_{max}}$ denoting the domain and range of
operator $A_{max}$ adjoint to $A_{min}$). Assume that $A_{min}$ is completely nonselfadjoint (simple),
i.e., there exists no reducing subspace $H_0$ in $H$ such that the
restriction $A_{min}|H_0$ is a selfadjoint operator in $H_0.$ Further
assume that the deficiency indices of $A_{min}$ (probably
being infinite) are equal: $n_+(A_{min})=n_-(A_{min})\le\infty.$

\begin{defn}[\cite{Gor,Ko1,DM}]\label{BT_def}
Let $\Gamma_0,\ \Gamma_1$ be linear mappings of $D_{A_{max}}$ to
$\mathcal{H}$ -- a separable Hilbert space. The triple $(\mathcal{H},
\Gamma_0,\Gamma_1)$ is called \emph{a boundary triple}
for the operator $A_{max}$ if:
\begin{enumerate}
\item for all $f,g\in D_{A_{max}}$
$$
(A_{max} f,g)_H -(f, A_{max}
g)_H = (\Gamma_1 f, \Gamma_0 g)_{\mathcal{H}}-(\Gamma_0 f,
\Gamma_1 g)_{\mathcal{H}}.
$$
\item the mapping $\gamma$ defined as $f\longmapsto (\Gamma_0 f; \Gamma_1
f),$ $f\in D_{A_{max}}$ is surjective, i.e.,
for all $Y_0,Y_1\in\mathcal{H}$ there exists such $y\in
D_{A_{max}}$ that $\Gamma_0 y=Y_0,\ \Gamma_1 y =Y_1.$
\end{enumerate}
\end{defn}

A boundary triple can be constructed for any operator $A_{min}$ of the class considered. Moreover, the
space $\mathcal H$ can be chosen in a way such that $\dim \mathcal H=n_+=n_-.$

\begin{defn}[\cite{Gor,DM}]
A nontrivial extension ${A}_B$ of the operator $A_{min}$ such that $A_{min}\subset  A_B\subset A_{max}$  is
called \emph{almost solvable} if there exists a boundary triple $(\mathcal{H},
\Gamma_0,\Gamma_1)$ for $A_{max}$ and a bounded linear operator $B$ defined everywhere on
$\mathcal{H}$ such that for every $f\in D_{A_{max}}$
$$
f\in D_{A_B}\text{ if and only if } \Gamma_1 f=B\Gamma_0 f.
$$
\end{defn}

It can be shown that if an extension $A_B$ of $A_{min}$, $A_{min}\subset  A_B\subset A_{max}$, has regular
points (i.e., the points belonging to the resolvent set)
in both upper and lower half-planes of the complex plane, then this extension is almost solvable.

The following theorem holds:
\begin{thm}[\cite{Gor,DM}]\label{old-extra}
Let $A_{min}$ be a closed densely defined symmetric operator with $n_+(A_{min})=n_-(A_{min}),$
let $(\mathcal{H},\Gamma_0,\Gamma_1)$ be a boundary triple of $A_{max}$.
Consider the almost solvable extension
$A_B$ of $A_{min}$ corresponding to the bounded operator $B$
in $\mathcal{H}.$ Then:
\begin{enumerate}
\item $y\in D_{A_{min}}$ if and only if $\Gamma_0 y=\Gamma_1 y=0,$
\item $ A_B$ is maximal, i.e., $\rho( A_B)\not=\emptyset$,
\item $(A_B)^*\subset A_{max},\ (A_B)^*=
A_{B^*},$
\item operator $A_B$ is dissipative if and only if $B$
is dissipative,
\item $(A_B)^*=A_B$ if and only if $B^*=B.$
\end{enumerate}
\end{thm}

The generalized Weyl-Titchmarsh M-function is then defined as follows.
\begin{defn}[\cite{DM,Gor,Koch}]\label{M-def}
Let $A_{min}$ be a closed densely defined symmetric operator,
$n_+(A_{min})=n_-(A_{min}),$ $(\mathcal{H},\Gamma_0,\Gamma_1)$ is its
space of boundary values. The operator-function $M(\lambda),$
defined by
\begin{equation}\label{Weyleq}
M(\lambda)\Gamma_0 f_{\lambda}=\Gamma_1 f_{\lambda},
\ f_{\lambda}\in \ker (A_{max}-\lambda),\  \lambda\in
\mathbb{C}_\pm,
\end{equation}
is called the Weyl-Titchmarsh M-function of a symmetric operator $A_{min}.$
\end{defn}

The following Theorem describing the properties of the M-function clarifies its meaning.
\begin{thm}[\cite{Gor,DM}, in the form adopted in \cite{RyzhovOTAA}]\label{Weyl}
Let $M(\lambda)$ be the M-function of a symmetric operator
$A_{min}$ with equal deficiency indices ($n_+(A_{min})=n_-(A_{min})<\infty$).
Let $A_B$ be an almost solvable extension of
$A_{min}$ corresponding to a bounded operator $B.$ Then for every $\lambda\in
\mathbb{C}:$
\begin{enumerate}
\item $M(\lambda)$ is analytic operator-function when
$Im\lambda\not=0$, its values being bounded linear operators
in $\mathcal{H}.$
\item $(Im\ M(\lambda))Im\ \lambda>0$ when $Im \lambda\not =0.$
\item $M(\lambda)^*=M(\overline{\lambda})$ when $Im \lambda\not =0.$
\item $\lambda_0\in \rho(A_B)$ if and only if $(B-M(\lambda))^{-1}$ admits bounded analytic continuation into the point $\lambda_0$.
\end{enumerate}
\end{thm}

In view of the last Theorem, one is tempted to reduce the study of the spectral properties of the Laplacian on a quantum graph to the study
of the corresponding Weyl-Titchmarsh M-function. Indeed, if one considers the operator under investigation as an extension of a properly chosen
symmetric operator defined on the same graph and constructs a boundary triple for the latter, one might utilize all the might of the complex analysis and the
theory of analytic matrix R-functions, since in this new setting the (pure point) spectrum of the quantum Laplacian is located exactly at the points into which
the matrix-function $(B-M(\lambda))^{-1}$ cannot be extended analytically (wagely speaking, these are ``zeroes'' of the named matrix-function).

It might appear as if the non-uniqueness of the space of boundary values and the resulting non-uniqueness of the Weyl-Titchmarsh M-function leads to
some problems on this path; but on the contrary, this flexibility of the apparatus is an advantage of the theory rather than its weakness. Indeed, as we are
going to show below, this allows us to ``separate'' the data describing the metric graph (this information will be carried by the M-function) from the data
describing the matching conditions at the vertices (this bit of information will be taken care of by the matrix $B$ parameterizing the extension). In turn,
this ``separation'' proves to be quite fruitful in view of applications that we have in mind.

There is yet another question to be taken care of along the way. As mentioned above, in order to make the approach suggested work one must ensure that the symmetric
operator $A_{min}$ is simple, i.e., does not have self-adjoint ``parts''. If it so happens that this operator looses simplicity (as we will show below, this
certainly happens if the graph contains loops and might happen if it contains cycles), one then ends up with the matrix-function $B-M(\lambda)$ which no longer
carries all the information about the spectrum of the corresponding quantum Laplacian. Namely, all the (point) spectrum of the self-adjoint ``part'' of the symmetric
operator $A_{min}$ will be invisible for this matrix-function.

Although as it is easily seen this is hardly a problem from the point of view of the present paper, it might complicate the issue when investigating other kinds
of direct and inverse spectral problems. It is due to this reason that we have elected to cover the question of simplicity in some details in the present paper (see
Theorems \ref{simplicity} and \ref{simplicity-p} towards the end of this Section).

We proceed with an explicit construction of the ``natural'' boundary triple and M-function in the case of graph Laplacians.

\subsubsection*{Construction of a boundary triple and calculation of the M-function in the case of a quantum Laplacian}
Let $\Gamma$ be a fixed finite compact metric graph. Let us denote by $\partial\Gamma$ the graph boundary, i.e., all the vertices of the graph
which have valence 1. W.l.o.g. we further assume that at all the vertices the matching conditions are of $\delta$ type (the case when they are of
$\delta'$ type is treated along the same lines and we provide the corresponding results without a proof; the mixed case can be looked at in more or less
the same fashion; we omit any discussion of the latter in order to improve readability of the paper).

As the operator $A_{max}$ rather then $A_{min}$ is crucial from the point of view of construction of a boundary triple, we start with this maximal operator and explicitly
describe its action and domain: $A_{max}= - \frac{d^2}{dx^2}$,
\begin{equation}\label{DAmax}
D(A_{max})=\left\{ f
\in \bigoplus_{j=1}^n W^2_2(\Delta_j)\ |\ \forall\ V_m \in V_{\Gamma
\backslash
\partial
\Gamma} f \text{ is continuous at }
V_m \right\}.
\end{equation}
\begin{rem}
Note that the operator chosen is not the ``most maximal'' maximal one: one could of course skip the condition of continuity through internal vertices;
nevertheless, the choice made proves to be the most natural from the point of view expressed above. This is exactly due to the fact that the graph connectivity is thus
reflected in the domain of the maximal operator and therefore propels itself into the expression for the M-matrix. Moreover, it should be noted that
this choice is also natural since the dimension of the M-matrix will be exactly equal to the number of graph vertices.
\end{rem}

The choice of the operators $\Gamma_0$ and $\Gamma_1$, acting onto
$\mathbb{C}_N$, $N = |V_\Gamma|$ is made as follows (cf., e.g., \cite{Aharonov} where a similar choice is suggested, but only for the graph boundary):
\begin{equation}\label{BT_formula}\Gamma_0f=\left( \begin{array}{c} f(V_1)\\f(V_2)\\ \ldots\\
f(V_N)
\end{array}\right);\quad
\Gamma_1f=\left( \begin{array}{c} \sum\limits_{x_j:x_j \in V_1} \partial_n f(x_j)\\\sum\limits_{x_j:x_j \in V_2} \partial_n f(x_j)\\ \ldots\\
\sum\limits_{x_j:x_j \in V_N} \partial_n f(x_j)
\end{array}\right).\end{equation}
Here the symbol $f(V_j)$ denotes the value of the function $f(x)$ at the vertex $V_j$. The latter is meaningful because of the choice of the domain of the maximal operator.

\begin{rem}
If one ascertains that the triple $(\mathbb{C}_N, \Gamma_0, \Gamma_1)$ satisfies Definition \ref{BT_def}, the corresponding minimal operator $A_{min}$ will therefore
be the following one: $A_{min}= - \frac{d^2}{dx^2}$,
\begin{multline}\label{Amin}
D(A_{min})=\left\{ f \in \bigoplus_{j=1}^n W^2_2(\Delta_j)\ |\
\forall\ V_m, m=1,...,N, \ f(V_m)=0,\right.\\
\left.\forall\ V_m, m=1,\ldots\, N,\  \sum_{x_j \in V_m}
\partial_n f(x_j)=0\right\}. \end{multline}
This operator will be symmetric with deficiency indices $(N,N)$. This follows from the fact that the domain of the minimal operator
admits the following characterization in terms of boundary triples: $D(A_{min})=\{f\in\ D(A_{max})|\Gamma_0 f=\Gamma_1 f=0\}$ (see Theorem \ref{old-extra}).
\end{rem}

\begin{lem}\label{BT_proof}
The triple $(\mathbb{C}_N; \Gamma_0, \Gamma_1)$, $N =
|V_\Gamma|$ is a boundary triple for the operator $A_{max}$ in the sense of Definition \ref{BT_def}.
\end{lem}

\begin{proof}
First, we verify the abstract Green formula.  Indeed, performing double integration by parts,
\begin{multline*}
\langle A_{max} f, g\rangle - \langle f, A_{max}
g\rangle = \\
\sum_{j=1}^{n} \left[ -f(x_{2j})
\bar{g}' (x_{2j}) + f(x_{2j-1}) \bar{g}' (x_{2j-1}) + f'(x_{2j})
\bar{g} (x_{2j}) -\right.\\
\left. f'(x_{2j-1}) \bar{g} (x_{2j-1}) \right ]=
\sum_{k=1}^{2n} \left[ \partial_n f(x_{k}) \bar{g} (x_{k}) - f(x_{k}) \partial_n \bar{g} (x_{k})
 \right ],
\end{multline*}
where the definition of the normal derivative on the graph has been taken into account. Splitting the last sum into parts corresponding the graph vertices,
one arrives at:
\begin{multline*}
\langle A_{max} f, g\rangle - \langle f, A_{max}
g\rangle = \\
\sum_{i=1}^N
\sum_{k: x_k\in V_i} \partial_n f(x_{k}) \bar{g} (x_{k})
-\sum_{i=1}^N \sum_{k: x_k\in V_i} f(x_{k}) \partial_n \bar{g}
(x_{k})=\\
 \langle \Gamma_1 f, \Gamma_0 g\rangle_{\mathbb{C}_n} - \langle \Gamma_0 f, \Gamma_1 g\rangle_{\mathbb{C}_n},
\end{multline*}
as required.

It remains to be shown that the mapping $f \mapsto \Gamma_0 f \oplus \Gamma_1 f$, $f\in
D(A_{max})$ is surjective as a mapping onto $\mathbb{C}_n \oplus \mathbb{C}_n$.

All we need to do is to show that for a pair of arbitrary vectors
$y=(y_1,\ldots,y_N)$, $z=(z_1,\ldots,z_N)$ there exists a function
$f\in D(A_{max})$ such that $\Gamma_0 f = y$, $\Gamma_1 f=z$.

Consider the vertex $V_1$ of valence $v_1$. Fix some edge containing $V_1$ and denote it
$\gamma_1$. The rest of the edges containing $V_1$ will be numbered in some arbitrary order and denoted  $\gamma_2$,..., $\gamma_{v_1}$.
Put $\partial_n f_{\gamma_1}(V_1)=z_1$, $\partial_n
f_{\gamma_j}(V_1)=0$, $j=2,...,v_1$ and $f_{\gamma_j}(V_1)=y_1$,
$j=1,...,v_1$. Then both required conditions are satisfied: the function to be constructed is continuous through the vertex $V_1$,
whereas  $\sum\limits_{x_j \in
V_1}
\partial_n f(x_j) = z_1$. Now pick the remaining vertices one by one. We end up with the trivial task of finding a function belonging to $\oplus_{j=1}^n W_2^2(\Delta_j)$
such that on each individual interval $\Delta_j$ the values of the function itself and of its derivative at both endpoints are fixed to some
predetermined values.
\end{proof}

\begin{rem}
If one considers a graph Laplacian with matching conditions of $\delta'$ type, the choice of the maximal operator and the corresponding
boundary triple (an analogue of Lemma \ref{BT_proof} can be obtained along the same lines) has to change accordingly:
$A_{max}= - \frac{d^2}{dx^2}$,
\begin{gather}
D(A_{max})=\left\{ f
\in \bigoplus_{j=1}^{n} W^2_2(\Delta_j) | \forall\ V_m \in
V_{\Gamma \backslash
\partial
\Gamma}\ \partial_n f \text{ is continuous at } V_m \right\}\label{DAmaxp}\\
\Gamma_0 f=\left( \begin{array}{c} \partial_n f(V_1)\\ \partial_n f(V_2)\\ \ldots\\
\partial_n f(V_N)
\end{array}\right);\quad
\Gamma_1f=\left( \begin{array}{c} \sum\limits_{x_j:x_j \in V_1} f(x_j)\\\sum\limits_{x_j:x_j \in V_2} f(x_j)\\ \ldots\\
\sum\limits_{x_j:x_j \in V_N} f(x_j)
\end{array}\right).\label{BT_formula_p}
\end{gather}
Then the minimal operator $A_{min}= -
\frac{d^2}{dx^2}$ on the domain
\begin{multline}\label{Amin-p}
D(A_{min})=\left\{ f \in \bigoplus_{j=1}^{n} W^2_2(\Delta_j)\ |
\ \forall\ V_m,\, m=1,...,N \
\partial_n f(V_m)=0,\right.\\
\left.\forall\ V_m, m=1,\ldots\, N\  \sum_{x_j \in V_m}
f(x_j)=0\right\}
\end{multline}
is again a symmetric operator with deficiency indices equal to $(N,N)$.
\end{rem}

We are now ready to formulate our main result of the present section, namely, the formula for the Weyl-Titchmarsh M-function associated
with the boundary triple \eqref{BT_formula}.

\begin{thm}\label{structure-delta}
Let $\Gamma$ be a finite compact metric graph. Let the operator $A_{max}$ be the negative second derivative on the domain \eqref{DAmax}. Let
the boundary triple for $A_{max}$ be chosen as $(\mathbb{C}^{N}, \Gamma_0, \Gamma_1)$, where $N$ is the number of vertices of $\Gamma$ and the operators
$\Gamma_0$ and $\Gamma_1$ are defined by \eqref{BT_formula}. Then the generalized Weyl-Titchmarsh M-function is an $N\times N$ matrix with matrix
elements given by the following formula.
\begin{equation}\label{M-delta}
m_{jp}=\begin{cases}
{\scriptstyle - k \sum\limits_{\Delta_t \in E_j} \cot{kl_t}+2k \sum\limits_{\Delta_t \in L_j} \tan{\frac{kl_t}{2}},}& {\scriptstyle j=p}, \\
{\scriptstyle k \sum\limits_{\Delta_t \in C_{j,p}} \frac{1}{\sin{kl_t}},}& {\scriptstyle j\neq p, \text{{\scriptsize vertices }} V_j \text{{\scriptsize and }} V_p}\\
& \text{{\scriptsize are connected by an edge}},\\
{\scriptstyle 0},& {\scriptstyle j\neq p,\, \text{{\scriptsize
vertices}}\, V_j\, \mbox{{\scriptsize and}}\, V_p}
\\
&\text{{\scriptsize are not connected by an edge}}.
\end{cases}
\end{equation}
Here $k=\sqrt{\lambda}$ (the branch of the square root is fixed so that $\text{Im } k\geq 0$), $E_j$ is the set of the graph edges
such that they are not loops and one of their endpoints belongs to the vertex
$V_j$, $L_j$ is the set of the loops attached to the vertex $V_j$, and finally, $C_{j,p}$ is the set of all graph edges which have
both $V_j$ and $V_p$ as endpoints (i.e., graph edges connecting vertices $V_j$ and $V_p$).
\end{thm}

\begin{proof}
The proof is an explicit calculation.

Consider the set of functions $f^\lambda \in Ker(A_{max} -
\lambda I)$. Clearly, on each edge $\Delta_t$ of the graph $\Gamma$
the function $f^\lambda|_{\Delta_t}$ is of the form $a_t^+ e^{ikx} + a_t^-
e^{-ikx}$, where $a_t^+$ and $a_t^-$ are some constants. These constants are chosen in a way such that
the function $f^\lambda$ is continuous through every internal vertex of the graph.

By definition of the Weyl-Titchmarsh M-matrix (see Definition \ref{M-def}), the identity $M(\lambda)\Gamma_0
f^\lambda= \Gamma_1 f^\lambda$  has to hold for  all $ f^\lambda$ such that
$f^\lambda \in Ker(A_{max} - \lambda I)$.

Consider a vertex $V_j$ having valence $v_j$ and check that $$M^j(k) \Gamma_0 f^\lambda = (\Gamma_1
f^\lambda)_j,$$ where $M^j(k)$ is the $j$-th row of the matrix $M(k)$ given by the formula \eqref{M-delta}. Since
$\Gamma_0 f_\lambda = (f^\lambda(V_1),\,f^\lambda(V_2),\, ...
,\,f^\lambda(V_N))$, we immediately obtain:
\begin{multline}\label{mult1}
M^j(k) \Gamma_0 f_\lambda = \left[- k \sum\limits_{\Delta_t \in E_j} \cot{kl_t}+2k \sum\limits_{\Delta_t \in L_j}
\tan{\frac{kl_t}{2}}\right] f^\lambda(V_j) + \\  k
\sum_{p: C_{j,p}\not=\varnothing}\sum\limits_{\Delta_t \in C_{j,p}} \frac{1}{\sin{kl_t}}
f^\lambda_{\Delta_t}(V_p),
\end{multline}
where $f^\lambda_{\Delta_t}:=f^\lambda|_{\Delta_t}$.

Note that in our notation $\cup_{p: C_{j,p}\not=\varnothing} C_{j,p}=E_j$. Moreover, due to continuity of the function $f^\lambda$ through the
vertex $V_j$ one has: $f^\lambda(V_j)=f_{\Delta_t}^\lambda(V_j)$ for all $t: \Delta_t\in E_j$. This gives ground to the separate consideration of terms in the last sum,
related to each particular edge $\Delta_t$ and connecting the vertex $V_j$ with a vertex $V_p$ for any admissible $p$ (for
the moment we shift our attention away from the loops attached to $V_j$, thus $p\not=j$). If the vertex $V_j$
is the left endpoint of the edge $\Delta_t=[0,l_t]$ and the vertex $V_p$ is the right one, we obtain:
\begin{multline*}
\frac{1}{\sin{(kl_t)}}f^\lambda_{\Delta_t}(V_p)-\cot{(kl_t)}
f^\lambda_{\Delta_t}(V_j)=\\
\frac{1}{\sin{(kl_t)}}\left( a_{\Delta_t}^+ [\exp{(ikl_t)}-\cos{(kl_t)}]+a_{\Delta_t}^- [\exp{(-ikl_t)}-\cos{(kl_t)}]\right)=\\
i\left( a_{\Delta_t}^+ - a_{\Delta_t}^-\right)=i {f^\lambda_{\Delta_t}}'(0)=i \partial_n f^\lambda_{\Delta_t}(V_j).
\end{multline*}
If on the other hand $V_j$ is the right endpoint of the edge
$\Delta_t=[0,l_t]$, $V_p$ being the left one, then
\begin{multline*}
\frac{1}{\sin{(kl_t)}}f^\lambda_{\Delta_t}(V_p)-\cot{(kl_t)}
f^\lambda_{\Delta_t}(V_j)=\frac{1}{\sin{(kl_t)}}f^\lambda_{\Delta_t}(0)-\cot{(kl_t)}
f^\lambda_{\Delta_t}(l_t)=\\
\frac{1}{\sin{(kl_t)}}\left( a_{\Delta_t}^+ [1-\exp{(ikl_t)}\cos{(kl_t)}]+a_{\Delta_t}^- [1-\exp{(-ikl_t)}\cos{(kl_t)}]\right)=\\
-i a_{\Delta_t}^+ \exp{(ikl_t)} + a_{\Delta_t}^-\exp{(-ikl_t)}=- i {f^\lambda_{\Delta_t}}'(l_t)=i \partial_n f^\lambda_{\Delta_t}(V_j).
\end{multline*}

If, finally, a loop $\Delta=[0,l]$ is attached to the vertex $V_j$, the set $L_j$ is non-empty and the sum over $L_j$ in \eqref{mult1} gives us
the corresponding term of the form $2k \tan{\frac{kl}{2}}$. Then
\begin{multline*}
2k \tan{\frac{kl}{2}} f^\lambda_{\Delta}(V_j)=k \tan{\frac{kl}{2}} \left[ f^\lambda_{\Delta}(0)+ f^\lambda_{\Delta}(l) \right]= \\
k \tan{\frac{kl}{2}} \left[ \alpha^+_\delta (1+\exp{(ikl)})+ \alpha^-_\delta (1+\exp{(-ikl)}) \right]= \\
-i k \frac{\exp{(i\frac{kl)}2)}-\exp{(i\frac{-kl_t}2)}}{\exp{(i\frac{kl}2)}+\exp{(-i\frac{kl}2)}}\left[ \alpha^+_\delta (1+\exp{(ikl)})+ \alpha^-_\delta (1+\exp{(-ikl)}) \right]=\\
-ik \left[ \alpha^+_\delta \left(1+\exp{(ikl)}\right)\frac{\exp{(ikl)}-1}{\exp{(ikl)}+1}+\right.\\ \left.\alpha^-_\delta \left(1+\exp{(-ikl_t)}\right)\frac{1-\exp{(-ikl)}}{1+\exp{(-ikl)}} \right]=\\
ik \left[ \alpha^+_\delta (1-\exp{(ikl)})- \alpha^-_\delta (1-\exp{(-ikl)}) \right]= i\left( {f^\lambda_{\Delta}}'(0)-{f^\lambda_{\Delta}}'(l)\right).
\end{multline*}

Thus we have ascertained that
$$M^j(k) \Gamma_0 f^\lambda = ik \sum \partial_n
f^\lambda_{\Delta_t}(V_j)= (\Gamma_1 f^\lambda)_j,$$
where the sum in the last formula is taken over all edges coming into or out of the vertex $V_j$.

Since $j=1,\dots, N$ is arbitrary, this completes the proof.
\end{proof}

\begin{eg}
Suppose the graph $\Gamma$ is
\begin{center}
\includegraphics[width=.95\textwidth]{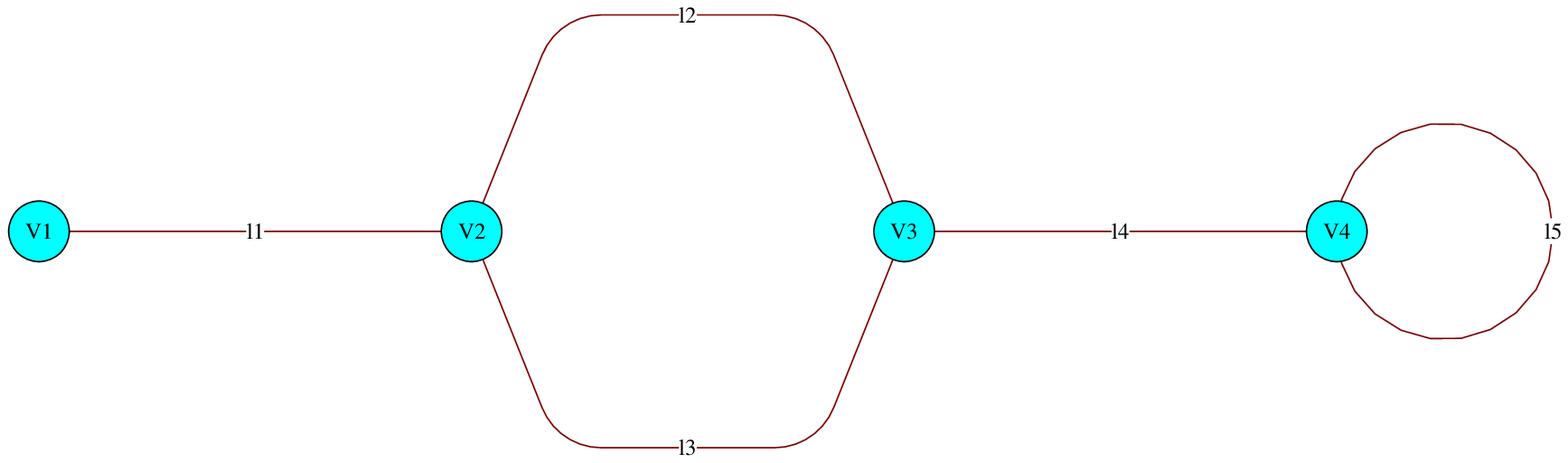}
\end{center}
Then the Weyl-Titchmarsh M-function of Theorem \ref{structure-delta} for this graph has the following form.
$$\footnotesize \left(\begin{matrix}
-k \cot (kl_1)&\frac{k}{\sin (kl_1)}&0&0\\
\frac{k}{\sin (kl_1)}&-k \sum\limits_{t=1}^3\cot (kl_t)&k \sum\limits_{t=2}^3\frac{1}{\sin (kl_t)}&0\\
0&k \sum\limits_{t=2}^3\frac{1}{\sin (kl_t)}&-k \sum\limits_{t=2}^4\cot (kl_t)&\frac{k}{\sin (kl_4)}\\
0&0&\frac{k}{\sin (kl_4)}&-k \cot (kl_4)+2k\tan(\frac{kl_5}2)
\end{matrix}\right)
$$
\end{eg}

A few remarks are in order.
\begin{rem}
As follows from the proof given, the Weyl-Titchmarsh M-function in our setting does not depend on the directions of graph edges, i.e., the M-function stays the same
if on any of the graph edges the left endpoint and the right endpoint swap places. This effect is of course well in line with what is well-known about spectra
of quantum graphs, see e.g., \cite{Kura2,Kuchment}.
\end{rem}

\begin{rem}
Provided that the graph has no loops, the value of M-function at zero, $M(0):=\lim_{\lambda\to 0}M(\lambda)$, turns out to be equal to the adjacency matrix
$C_\Gamma$ of the metric graph defined
in the following way:
$$
\{C_\Gamma\}_{jp}:=\begin{cases}
\sum_{\Delta_t \in E_j} \frac 1{l_t}, &j=p\\
\sum_{\Delta_t \in C_{j,p}} \frac 1{l_t}, &j\not=p.
\end{cases}
$$
This adjacency matrix in the special case when all the edges have unit lengths is exactly the sum of the classical adjacency matrix $A_\Gamma$ and the diagonal matrix
of vertex valences, where $A_\Gamma$ is defined as follows:
$$
\{A_\Gamma\}_{jp}:=\begin{cases}
0, &j=p\\
\sum_{\Delta_t \in C_{j,p}} 1, &j\not=p.
\end{cases}
$$
Thus one might convince oneself that the information on the connectivity of the graph is actually represented in the M-function (w.r.t.
the boundary triple used by us) in a very transparent way.
\end{rem}

In the situation when matching conditions at all the graph vertices are of $\delta'$ type (and the maximal operator $A_{max}$ and the
boundary triple for is are chosen accordingly) the following result can be easily obtained along the
same lines.
\begin{thm}\label{structure-delta-p}
Let $\Gamma$ be a finite compact metric graph. Let the operator $A_{max}$ be the negative second derivative on the domain \eqref{DAmaxp}. Let
the boundary triple for $A_{max}$ be chosen as $(\mathbb{C}^{N}, \Gamma_0, \Gamma_1)$, where $N$ is the number of vertices of $\Gamma$ and the operators
$\Gamma_0$ and $\Gamma_1$ are defined by \eqref{BT_formula_p}. Then the generalized Weyl-Titchmarsh M-function is an $N\times N$ matrix with matrix
elements given by the following formula.
$$m_{jp}=\begin{cases}
{\scriptstyle \frac{1}{k} \sum\limits_{\Delta_t \in E_j} \cot{(kl_t)}+2 \frac{1}{k} \sum\limits_{\Delta_t \in L_j} \cot{\frac{kl_t}{2}},}& {\scriptstyle j=p}, \\
{\scriptstyle k \sum\limits_{\Delta_t \in C_{j,p}} \frac{1}{\sin{kl_t}},}& {\scriptstyle j\neq p, \text{{\scriptsize vertices }} V_j \text{{\scriptsize and }} V_p}\\
& \text{{\scriptsize are connected by an edge}},\\
{\scriptstyle 0},& {\scriptstyle j\neq p,\, \text{{\scriptsize
vertices}}\, V_j\, \mbox{{\scriptsize and}}\, V_p}
\\
&\text{{\scriptsize are not connected by an edge}}.
\end{cases}
$$
Here $k=\sqrt{\lambda}$ (the branch of the square root is fixed so that $\text{Im } k\geq 0$), $E_j$ is the set of the graph edges
such that they are not loops and one of their endpoints belongs to the vertex
$V_j$, $L_j$ is the set of the loops attached to the vertex $V_j$, and finally, $C_{j,p}$ is the set of all graph edges which have
both $V_j$ and $V_p$ as endpoints (i.e., graph edges connecting vertices $V_j$ and $V_p$).
\end{thm}

Now the following obvious statement demonstrates that the choice of maximal operators made in \eqref{DAmax} (\eqref{DAmaxp}, respectively) and
of boundary triples made in \eqref{BT_formula} (\eqref{BT_formula_p}, respectively) is indeed natural for the study of spectral properties of
quantum Laplacians with matching conditions of $\delta$ type ($\delta'$ type, respectively).

\begin{lem}
(i) A quantum Laplacian with $\delta$-type matching conditions in the sense of Definition \ref{def_Laplace} is an almost solvable extension of the symmetric operator
$A_{min}=A_{max}^*$, where $A_{max}$ is defined by \eqref{DAmax}, w.r.t. the boundary triple $(\mathbb{C}^N,\Gamma_0, \Gamma_1)$ with $\Gamma_0$ and $\Gamma_1$ defined
by \eqref{BT_formula}. Its parameterizing matrix $B$ w.r.t. this boundary triple is diagonal, $B=\text{diag}(\alpha_1,\dots,\alpha_N)$, where $\{\alpha_k\}_{k=1}^N$
are the coupling constants of Definition \ref{def_matching}.\\
(ii) A quantum Laplacian with $\delta'$-type matching conditions in the sense of Definition \ref{def_Laplace} is an almost solvable extension of the symmetric operator
$A_{min}=A_{max}^*$, where $A_{max}$ is defined by \eqref{DAmaxp}, w.r.t. the boundary triple $(\mathbb{C}^N,\Gamma_0, \Gamma_1)$ with $\Gamma_0$ and $\Gamma_1$ defined
by \eqref{BT_formula_p}. Its parameterizing matrix $B$ w.r.t. this boundary triple is diagonal, $B=\text{diag}(\alpha_1,\dots,\alpha_N)$, where $\{\alpha_k\}_{k=1}^N$
are the coupling constants of Definition \ref{def_matching}.
\end{lem}

It follows now from Theorem \ref{Weyl} out that at least \emph{a part} of the spectrum of a quantum Laplacian with $\delta$ or $\delta'$ matching conditions can be characterized in terms
of the $N\times N$ analytic matrix R-function $B-M(\lambda)$, where $B$ is the diagonal matrix of coupling constants and $M(\lambda)$ is the corresponding
Weyl-Titchmarsh M-function. Moreover, provided that the corresponding minimal operator $A_{min}$ is simple, i.e., has no reducing self-adjoint parts, \emph{all} of the
spectrum of the quantum Laplacian can be characterized this way. It turns out that in our situation we are able to give a criterion of when this happens.
\begin{thm}\label{simplicity} Suppose that $\Gamma$ is a finite compact metric graph. Let the operator $A_{max}$ be the negative second derivative on the domain \eqref{DAmax}. Let
$A_{min}=A_{max}^*$ (the domain of $A_{min}$ is then described by \eqref{Amin}). Then the symmetric operator $A_{min}$ is simple if and only if
(i) the graph $\Gamma$ has no loops \emph{and} (ii) every cycle belonging to the graph $\Gamma$ has rationally independent edge lengths.
\end{thm}
In order to carry out the proof of this Theorem, we start with the following almost obvious Lemma.
\begin{lem}\label{simple-lemma}
In the setting of the preceding Theorem, the operator $A_{min}$ is simple if and only if it has no real eigenvalues.
\end{lem}
\begin{proof}[Proof of Lemma]
Suppose first that the operator $A_{min}$ has an eigenvalue $\lambda_0$ with an associated eigenfunction $\phi_0$. Then the subspace generated by $\phi_0$ is
invariant for $A_{min}$ and hence reducing \cite{Birman}. It follows immediately that the operator $A_{min}$ is not simple.

On the other hand, let $A_{min}$ have no eigenvalues and suppose that it has a reducing subspace $H_0$, the restriction of the operator onto which is self-adjoint. Then
this subspace will be necessarily reducing for every extension of the operator $A_{min}$, in particular, for the self-adjoint operator of Dirichlet decoupling $A_D$
defined as the negative second derivative on the following domain:
$$
D(A_D)=\{f \in
\bigoplus_{j=1}^n W^2_2(\Delta_j) | \forall x_k,\, k=1,...,2n \quad
f(x_k)=0\}.
$$
Moreover, since $A_{min}|H_0$ is already self-adjoint by assumption, the following equality holds: $A_{min}|H_0=A_D|H_0$.
The operator $A_D$ is equal to the orthogonal sum over all the graph edges of regular Sturm-Liouville operators with Dirichlet boundary conditions,
\begin{equation}\label{D-decoupling}
\begin{gathered}
A_D=\oplus_{j=1}^n A_D(\Delta_j), \text{ where } A_D(\Delta_j)=-\frac {d^2}{dx^2} \text{ on }\\ D(A_D(\Delta_j))
=\{f\in W_2^2(\Delta_j)| f(x_{2j-1})=f(x_{2j})=0\}.
\end{gathered}
\end{equation}
It follows that $A_D$ (and thus $A_D|H_0$) has purely discrete spectrum. Therefore
we arrive at the conclusion that the operator $A_{min}|H_0$ ought to have at least one real eigenvalue, and thus the same applies to $A_{min}$.
The contradiction we have arrived to completes
the proof.
\end{proof}

We are now all set to continue with the proof of Theorem \ref{simplicity}.
\begin{proof}
We first prove that if the graph $\Gamma$ has no cycles and if every cycle belonging to it has rationally independent edge lengths, then the operator
$A_{min}$ has no real eigenvalues. Assume the opposite. Let $\lambda_0$ be its eigenvalue and $\phi_0$ be the associated eigenfunction.

First, we will show that $\phi_0$ cannot be supported by a tree (in the case when $\Gamma$ is a tree graph, this will complete the proof). Indeed, let $\Gamma_0
\subset \Gamma$ be a tree and suppose that $\phi_0$ is supported by $\Gamma_0$. Since on every edge $\Delta_j=[
x_{2j-1},x_{2j}]$ of $\Gamma$ not belonging to $\Gamma_0$
the eigenfunction $\phi_0$ is identically equal to zero and thus satisfies boundary conditions $\phi_0(x_{2j-1})=\phi_0(x_{2j})=\phi'_0(x_{2j-1})=\phi'_0(x_{2j})=0$,
on $\Gamma_0$ the function $\phi_0$ ought to satisfy the boundary conditions \eqref{Amin} as long as it satisfies them on the larger graph $\Gamma$.
Now pick any boundary vertex $V_k\in\partial \Gamma_0$ (i.e., a vertex having valence 1). At this vertex the function $\phi_0$ together with its first derivative
must therefore be zero, from where it follows immediately that the edge leading to the vertex $V_k$ does not support $\phi_0$.

The same applies to all vertices forming the graph boundary and to all the edges leading to them. As these do not support the function $\phi_0$, one may then drop
them altogether, which leads to a smaller graph $\tilde \Gamma_0\subset \Gamma_0$, which is still a tree. The procedure of trimming  the tree graph $\Gamma_0$
can be repeated as many times, as required. Since $\Gamma_0$ is a tree by assumption, after some finite number of iterations we are left with a graph with no edges.

Having established the fact that $\phi_0$ cannot be supported by a tree subgraph of $\Gamma$, we immediately obtain that it must be supported by at least one cycle
belonging to $\Gamma$. Indeed, if this is not so, $\phi_0$ must be supported by a tree or a collection of trees leading to an immediate contradiction.

Now pick a cycle $\Gamma_1\subset \Gamma$ which belongs to the support of $\phi_0$.
The function $\phi_0$ has to be equal to $\sin (\sqrt{\lambda_0} x)$ on each edge $\Delta_j=[0,l_j]$
(shifting as before w.l.o.g. the edge $\Delta_j$ so that its left endpoint is at zero) forming this cycle as the solution of the differential
equation $-\phi''_0=\lambda_0 \phi_0$ with zero boundary condition at the left endpoint. It is clear now that in order for the
non-trivial (i.e., supported by all edges of $\Gamma_1$) function $\phi_0$ to
be equal to zero at all the right endpoints of the edges forming $\Gamma_1$ it is necessary for the lengths of these edges to be rationally dependent.

Repeating this argument for every cycle of $\Gamma$ we arrive at the contradiction sought.

The proof of the inverse implication is by explicit construction. Indeed, in order to show that $A_{min}$ on a graph $\Gamma$ containing
a cycle with rationally dependent edge lengths has an eigenvalue, one simply constructs an eigenfunction supported solely by this cycle. On every edge $\Delta_j$ is has to be
equal to $\sin(\sqrt{\lambda_0}x)$. The existence of such non-trivial function is guaranteed by the fact that the edge lengths are rationally dependent.
The case of a loop is treated analogously.
\end{proof}

\begin{rem}
In terms of the operator of Dirichlet decoupling $A_D$ defined in \eqref{D-decoupling} it is easy to see that eigenvalues of $A_{min}$ (if any) might occur only
at points $(\frac{\pi m}{l_j})^2$, where $m \in \mathbb{Z}\setminus \{0\}$, $j=1,\dots,n$. Moreover, the eigenfunctions (if any) of $A_{min}$ are equal to those
eigenfunctions of $A_D$ which satisfy the matching conditions for the normal derivatives in \eqref{Amin}.
\end{rem}

\begin{rem}
If $\Gamma$ is a finite compact metric graph, the operator $A_{max}$ is the negative second derivative on the domain \eqref{DAmax} and
the boundary triple for $A_{max}$ is chosen as $(\mathbb{C}^{N}, \Gamma_0, \Gamma_1)$, where $N$ is the number of vertices of $\Gamma$ and the operators
$\Gamma_0$ and $\Gamma_1$ are defined by \eqref{BT_formula}, the generalized Weyl-Titchmarsh M-function has poles precisely at the points of the spectrum
of the operator $A_D$ of Dirichlet decoupling \eqref{D-decoupling} provided that the graph $\Gamma$ has no loops \emph{and} the edge lengths along every cycle of $\Gamma$ are
rationally independent.

The elementary proof of this is based on the explicit form of the M-function, see Theorem \ref{structure-delta}, and the work done in the proof of Theorem \ref{simplicity}
\end{rem}

If instead of the operator $A_{max}$ treated by Theorem \ref{simplicity} one considers the operator of the negative second derivative on $\Gamma$ defined on the
domain \eqref{DAmaxp}, Lemma \ref{simple-lemma}
continues to hold (with an elementary substitution of Dirichlet decoupling by the Neumann one).
Unfortunately, in this situation Theorem \ref{simplicity} fails. Instead, one can prove the following modification of it.
\begin{thm}\label{simplicity-p}
Suppose that $\Gamma$ is a finite compact metric graph. Let the operator $A_{max}$ be negative second derivative on the domain \eqref{DAmaxp}. Let
$A_{min}=A_{max}^*$ (the domain of $A_{min}$ is then described by \eqref{Amin-p}). Then the symmetric operator $A_{min}$ has no eigenvalues
\emph{away from zero} if and only if
(i) the graph $\Gamma$ has no loops \emph{and} (ii) every cycle belonging to the graph $\Gamma$ has rationally independent edge lengths.
\end{thm}
The \emph{proof} follows essentially the same lines as the proof of Theorem \ref{simplicity}. The only difference comes when one considers the candidate for
an eigenfunction on the cyclic part of the graph. On a cycle with an even number of edges, even despite the fact that the edge lengths are chosen to be rationally independent, one can construct
an eigenfunction of the operator $A_{min}$ corresponding to the point $\lambda=0$ by putting it to be equal to 1 on all odd edges and -1 on all even edges.

It follows that in the situation of graph Laplacians with $\delta'$ matching even the condition that the graph $\Gamma$ contains no loops and the edge lengths over all
cycles are rationally independent does not in general guarantee that the matrix-function $B-M(\lambda)$ carries all the spectral information about the extension
$A_B$. Nevertheless, it still carries full information about the spectrum away from zero.

\specialsection{Trace formulae for a pair of graph Laplacians}
In the present section, we apply the mathematical apparatus developed in Section 2 in order to study isospectral (i.e., having the same spectrum, counting
multiplicities) quantum Laplacians defined on a finite compact metric graph $\Gamma$. In order to do so, we will assume that the graph itself is given. Moreover,
we will assume that the matching conditions at all its vertices are of $\delta$ type ($\delta'$ type, respectively).

Considering a pair of such Laplacians which differ only in coupling constants defining the matching conditions we will derive an infinite series of trace formulae.

We proceed with our main theorem of this section.

\begin{thm}\label{trace}
Let $\Gamma$ be a finite compact metric graph having $N$ vertices. Let $A_{B_1}$ and $A_{B_2}$ be two graph Laplacians on the graph $\Gamma$ with $\delta$-type
matching conditions ($B_1=\text{diag}\{\tilde \alpha_1,\dots,\tilde \alpha_N\}$ and $B_2=\text{diag}\{ \alpha_1$, $\dots$, $\alpha_N\}$, where
both sets $\{\tilde \alpha_m\}$ and $\{\alpha_m\}$ are the sets of coupling constants in the sense of Definition \ref{def_matching}).
Let the (point) spectra of these two operators (counting multiplicities) be equal, $\sigma(A_{B_1}) = \sigma(A_{B_2})$.

Then the following infinite series of trace formulae holds:
$$\sum_{j=1}^{m}\frac1j C_{m-1}^{m-j} Tr (D^j
B_2^{m-j} \Gamma_N^{-m})=0, \quad m=1,2,\dots$$
where $D:=B_1-B_2$ and the matrix $\Gamma_N$ is the matrix of valences, $\Gamma_N=\text{diag}\{\gamma_1,\ldots, \gamma_N\}$, $\gamma_k$ being the valence of the vertex
$V_k$.
\end{thm}

\begin{proof}
We will use the apparatus developed in Section 2. Namely, we choose the maximal operator $A_{max}$ as in \eqref{DAmax}, the boundary triple \eqref{BT_formula} and
use the expression for the Weyl-Titchmarsh M-function of $A_{max}$ obtained in Theorem \ref{structure-delta}. Then w.r.t. the chosen boundary triple
the operators $A_{B_1}$ and $A_{B_2}$ are both almost solvable extensions of the operator $A_{min}=A_{max}^*$, parameterized by the matrices
$B_1$ and $B_2$, respectively. Throughout we of course assume that $D\not=0$.

We will now show that provided that the spectra of both given operators coincide, $\det
(B_1-M(\lambda))(B_2-M(\lambda))^{-1}\equiv 1$. This is done by a Liouville-like argument. Indeed, consider two matrix-functions
$M_j=(B_j-M(\lambda))\left(\frac{\sin
(\sqrt{\lambda} l_1)\sin (\sqrt{\lambda} l_2)\cdots\sin
(\sqrt{\lambda} l_n)}{(\sqrt{\lambda})^N}\right)$, $j=1,2$.
Put $F_j:=\det M_j$. Then, as can be easily seen from Theorem \ref{structure-delta}, $F_1$, $F_2$ are two scalar
analytic entire functions in $\mathbb C$. By Theorem \ref{Weyl} their fraction $F_1/F_2$ has no poles and no zeroes, since the spectra
of operators $A_{B_1}$ and $A_{B_2}$ coincide.

Now it can be easily ascertained that both $F_1$ and $F_2$ are of normal type and of order at least not greater than 1 \cite{Levin}. Then
their fraction is again an entire function of order not greater than 1 \cite{Levin}. Finally, by Hadamard's theorem $\frac{F_1}{F_2}=e^{a\lambda+b}$.

It remains to be seen that $a=b=0$. This follows immediately from the asymptotic behaviour of the matrix-function $M(\lambda)$ as $\lambda\to-\infty$.
Namely, $M(\sqrt{\lambda})=\sqrt{\lambda}A(\sqrt{\lambda})$ (see Theorem \ref{structure-delta}), where $A(\sqrt{\lambda})\rightarrow
i \Gamma_N$ as
$\lambda\rightarrow-\infty$. In fact, $A(\sqrt{\lambda})=i\Gamma_N +\bar o(\frac 1 {|\sqrt{\lambda}|^M})$ for any $M>0$, which essentially makes the
rest of the proof work.

We have thus obtained the following identity:
$$
1\equiv\det(B_1-M(\lambda))(B_2-M(\lambda))^{-1}=\det
(I+D(B_2-M(\lambda))^{-1}).
$$
Since the analytic matrix-function $(B_1-M(\lambda))(B_2-M(\lambda))^{-1}$ tends to $I$ as $\lambda\to-\infty$, it is can be diagonalized there.
We are then able to apply the standard formula connecting determinant and trace:
$$
\ln\det
(I+D(B_2-M(\lambda))^{-1})=Tr \ln (I+D(B_2-M(\lambda))^{-1}).
$$
Then
\begin{equation}\label{p1}
0= Tr \ln (I+D(B_2-M(\lambda))^{-1})=\sum_{j=1}^\infty \frac{(-1)^{j+1}}{j} Tr \left(D(B_2-M)^{-1}\right)^j.\end{equation}
The sum is absolutely convergent since $\left\|
(B_2-M(\lambda))^{-1} \right\|\ll 1$ as
$\lambda\rightarrow-\infty$, which again follows from the asymptotic behaviour of $M(\lambda)$ outlined above.

Consider $Tr (D(B_2-M(\lambda))^{-1})^j$. First, again using the explicit formula for $M(\lambda)$ obtained in Theorem \ref{structure-delta}, we note that
$$
B_2-M(\lambda)=B_2+\tau\Gamma_N+\bar o(\tau^{-M}) \text{ for arbitrarily large } M>0,
$$
where for the sake of convenience we have put $\tau:=-i \sqrt{\lambda}$ so that $\sqrt{\lambda}=i\tau$, $\tau\to+\infty$. Now from the second Hilbert
identity we immediately derive
$$
(B_2-M(\lambda))^{-1}=(B_2+\tau\Gamma_N)^{-1}+\bar o(\tau^{-M})
$$
for an arbitrarily large positive $M$. Then, clearly,
$$
(D(B_2-M(\lambda))^{-1})^j=(D(B_2+\tau\Gamma_N)^{-1})^j+\bar o(\tau^{-M}) \text { for all } j.
$$
Substituting this expression into \eqref{p1}, we have for an arbitrary large natural $M$:
$$
0= \sum_{j=1}^M \frac{(-1)^{j+1}}{j} Tr \left(D(B_2+\tau \Gamma_N)^{-1}\right)^j +\bar o(\tau^{-M}).
$$
Note, that all the matrices $D$, $B_2$ and $\Gamma_N$ are diagonal and thus commute. We will then expand $(B_2+\tau\Gamma_N)^{-j}$ into the power series
and substitute the result into the last formula. One has:
\begin{multline*}
(B_2+\tau\Gamma_N)^{-j}=\frac1{\tau^j}\left(I+\frac {\Gamma_N^{-1}B_2}{\tau}\right)^{-j} \Gamma_N^{-j}=\\
\frac1{\tau^j}\sum_{i=0}^\infty \frac1{\tau^i} C_{i+j-1}^i \Gamma_N^{-i}B_2^i \Gamma_N^{-j} (-1)^i=
\sum_{m=j}^\infty \frac1{\tau^m} C_{m-1}^{m-j} \Gamma_N^{-m} B_2^{m-j} (-1)^{m-j}=\\
\sum_{m=j}^M \frac1{\tau^m} C_{m-1}^{m-j} \Gamma_N^{-m} B_2^{m-j} (-1)^{m-j} + \bar o(\tau^{-M}).
\end{multline*}
The identity \eqref{p1} then yields:
\begin{multline}\label{p2}
0\equiv  Tr \ln (I+D(B_2-M(\lambda))^{-1})=\\
-\sum_{j=1}^M \frac 1j \sum_{m=j}^M \frac 1{\tau^m} C_{m-1}^{m-j}(-1)^m Tr (D^j \Gamma_N^{-m}B_2^{m-j})+\bar o(\tau^{-M})=\\
-\sum_{m=1}^M \frac {(-1)^m}{\tau^m}\sum_{j=1}^m \frac 1j C_{m-1}^{m-j} Tr (D^j\Gamma_N^{-m}B_2^{m-j})+\bar o(\tau^{-M}).
\end{multline}
Identity \eqref{p2} holds for any natural $M\gg 1$ and thus in the last sum each term of the form $\beta_m \tau^{-m}$ ought to be equal to zero.
This immediately yields the claim.
\end{proof}

Leaving the analysis of full countable set of trace formulae thus obtained for a forthcoming publication, we derive a few corollaries from
the last Theorem restricting consideration to just the first formula.

\begin{cor}
(i) Suppose that the matrices $B_1$ and $B_2$ are scalar (i.e., all the coupling constants in matching conditions coincide for the operators
$A_{B_1}$, $A_{B_2}$, respectively). Then if $\sigma(A_{B_1})=\sigma(A_{B_2})$, we obtain $B_1=B_2$. In other words, different graph Laplacians have
under the assumption made different spectra, or, to put it the other way around, the spectrum of graph Laplacian uniquely determines the coupling constants,
provided that all the coupling constants are equal.

(ii) If $B_1=0$ (which corresponds to the case of a graph Laplacian with standard, or Kirchhoff, matching conditions) and $B_2\geq 0$, the corresponding
operators $A_{B_1}$ and $A_{B_2}$ cannot have identical spectra.

(iii) If $B_1\geq B_2$ or $B_2\geq B_1$, the corresponding
operators $A_{B_1}$ and $A_{B_2}$ again cannot have identical spectra. Thus under the assumption that, roughly speaking, the
strength of matching condition is ordered,
the spectrum of graph Laplacian uniquely determines all the coupling constants.

(iv) If all the coupling constants in the matching conditions are known to be zero but for exactly one, the spectrum
of graph Laplacian uniquely determines  the operator.
\end{cor}

\begin{proof}
The first trace formula obtained in the last Theorem reads:
\begin{gather*}
Tr D \Gamma_N^{-1}=0
\end{gather*}
All the assertions follow immediately from this since $\Gamma_N>0$ and has no zero diagonal entries.
\end{proof}

The situation of graph Laplacian with $\delta'$ type matching conditions is similar, but somewhat different.
\begin{thm}\label{main2}
Let $\Gamma$ be a finite compact metric graph having $N$ vertices. Let $A_{B_1}$ and $A_{B_2}$ be two graph Laplacians on the graph $\Gamma$ with $\delta'$-type
 matching conditions ($B_1=\text{diag}\{\tilde \alpha_1,\dots,\tilde \alpha_N\}$ and $B_2=\text{diag}\{ \alpha_1$, $\dots, \alpha_N\}$, where
both sets $\{\tilde \alpha_m\}$ and $\{\alpha_m\}$ are the sets of coupling constants in the sense of Definition \ref{def_matching}).
Let the (point) spectra of these two operators (counting multiplicities) be equal, $\sigma(A_{B_1}) = \sigma(A_{B_2})$. Let further $B_1$ and $B_2$
be invertible.

Then the following infinite series of trace formulae holds:
$$\sum_{j=1}^{m}\frac1j C_{m-1}^{m-j} Tr (D^j
B_2^{-m+j} \Gamma_n^{j+m})=0, \quad m=1,2,\dots$$
where $D:=B_2^{-1}-B_1^{-1}$ and the matrix $\Gamma_N$ is the matrix of
valences, $\Gamma_N=\text{diag}\{\gamma_1,\ldots, \gamma_N\}$, $\gamma_k$ being the valence of the vertex
$V_k$.
\end{thm}
\begin{proof}[A sketch of the proof]
Certain minor technical differences compared to the proof of previous Theorem are due to the fact that in the case of $\delta'$ type matching conditions the
diagonal of the matrix $M(\lambda)$ decays as $\lambda\to-\infty$ instead of growing there. In order to cope with this situation, one considers
$B_1^{-1}(B_1-M(\lambda))$ instead of $B_1-M(\lambda)$ and $B_2^{-1}(B_2-M(\lambda))$ instead of $B_2-M(\lambda)$. This is possible since by assumption both matrices
$B_1$ and $B_2$ are invertible.

Then
\begin{multline*}
\det\left[B_1^{-1}(B_1-M(\lambda))(B_2^{-1}(B_2-M(\lambda)))^{-1}\right]=\\ \det(I+DM(\lambda)(I-B_2^{-1}M(\lambda))^{-1})
\end{multline*}
with the argument of determinant on the right having the required form of identity plus a vanishing term. The rest of the proof is carried along the same lines
as the proof of Theorem \ref{trace}.
\end{proof}

Due to the requirement that $B_1$ and $B_2$ are invertible, only the following two assertions based on the first trace formula remain valid in the situation
of graph Laplacians with $\delta'$ type matching conditions.
\begin{cor}
(i) Suppose that the matrices $B_1$ and $B_2$ are scalar (i.e., all the coupling constants in matching conditions coincide for the operators
$A_{B_1}$, $A_{B_2}$, respectively). Then if $\sigma(A_{B_1})=\sigma(A_{B_2})$, we obtain $B_1=B_2$.

(ii) If $B_1\geq B_2$ or $B_2\geq B_1$, the corresponding
operators $A_{B_1}$ and $A_{B_2}$ cannot have identical spectra.
\end{cor}

The corollaries derived above from Theorems \ref{trace} and \ref{main2} are formulated implicitly, i.e., they do not provide an explicit procedure of
reconstruction for the matrix $B$ based on the spectrum of the corresponding graph Laplacian $A_B$. Yet the approach suggested by us above
can be utilized in order to obtain, at least in some special cases, such procedures. We will discuss these elsewhere as in our view this discussion
is beyond the scope of the present paper.

\subsection*{Acknowledgements} The authors express their deep gratitude to Prof. Yurii Samojlenko and to Prof. Sergey Naboko for their constant attention
to the authors' work and for fruitful discussions. The second author also thanks the School of Mathematics of Cardiff University, where parts of the work were
done, and personally Prof. Marco Marletta for hospitality.

\end{document}